\documentclass{easychair}

\usepackage[utf8]{inputenc}
\usepackage{amsmath,amssymb,pxfonts}
\usepackage{todonotes}
\usepackage{xspace}
\usepackage{proof}
\usepackage{empheq}
\usepackage{etex}
\usepackage{breakurl}

\usepackage{enumitem}
\usepackage{multirow}
\usepackage{listings}
\usepackage[rightcaption]{sidecap}

\usepackage{comment}

\usepackage{url}
\usepackage{booktabs}

\usepackage{cleveref}
\crefname{section}{\S}{\S\S}

\usepackage{syntax}
\usepackage{fancyvrb}
\usepackage[noend]{algpseudocode}
\usepackage{algorithm}

\usepackage{etex}
\usepackage{pgfplots}
\usetikzlibrary{matrix}
\pgfplotsset{compat=newest}
\pgfplotsset{width=4*\textwidth/5}

\usepackage{doc}
\usepackage{makeidx}

\newcommand\refentry[1]{% positions two related legendimages in one cell
  \raisebox{1pt}{\ref{plot:#1ues-sup=off}}\llap{\raisebox{-1pt}{\ref{plot:#1ues-sup=on}}}%
}

\newtheorem{theorem}{Theorem}[section]

\newtheorem{lemma}[theorem]{Lemma}

\newtheorem{example}[theorem]{Example}
\newtheorem{definition}[theorem]{Definition}

\newcommand{\smt}[1]{\texttt{#1}}

\newcommand{\SMT}{SMT\xspace}
\newcommand{\SAT}{SAT\xspace}
\newcommand{\True}{\mathbf{true}}
\newcommand{\tuple}[1]{\left\langle #1 \right\rangle}
\newcommand{\set}[1]{\left\{ #1 \right\}}
\newcommand{\parens}[1]{\left( #1 \right)}
\newcommand{\alt}{\mathrel{|}}

\newcommand{\imp}{\Rightarrow}

\newcommand{\deq}{\triangleq}

\newcommand{\eventpo}{\ensuremath{\ll}\xspace}
\newcommand{\clockpo}{\ensuremath{\prec}\xspace}
\newcommand{\clockpoeq}{\ensuremath{\preceq}\xspace}

\newcommand{\events}{\ensuremath{\mathsf{E}}\xspace}
\newcommand{\reads}{\ensuremath{\mathsf{R}}\xspace}
\newcommand{\writes}{\ensuremath{\mathsf{W}}\xspace}

\newcommand{\clock}[1]{\ensuremath{\mathsf{c}_{#1}}\xspace}
\newcommand{\supremum}[1]{\ensuremath{\mathsf{sup}_{#1}}\xspace}
\newcommand{\readSel}[1]{\ensuremath{\mathsf{s}_{#1}}\xspace}
\newcommand{\writeSel}[1]{\ensuremath{\mathsf{s}_{#1}}\xspace}

\newcommand{\theoryClock}{\ensuremath{\mathcal{T}_{C}}\xspace}

\newcommand{\theoryValue}{\ensuremath{\mathcal{T}_{V}}\xspace}
\newcommand{\theorySelection}{\ensuremath{\mathcal{T}_{S}}\xspace}

\newcommand{\theoryInts}{\ensuremath{\mathcal{T}_{\mathbb{Z}}}\xspace}
\newcommand{\theoryReals}{\ensuremath{\mathcal{T}_{\mathbb{R}}}\xspace}
\newcommand{\theoryBV}{\ensuremath{\mathcal{T}_{\mathbb{BV}}}\xspace}

\newcommand{\valf}{\ensuremath{\textstyle val}\xspace}
\newcommand{\guardf}{\ensuremath{\textstyle guard}\xspace}

\newcommand{\readval}[1]{\ensuremath{\mathsf{rv}_{#1}}\xspace}
\newcommand{\writeval}[1]{\ensuremath{\valf(#1)}\xspace}
\newcommand{\guard}[1]{\ensuremath{\guardf(#1)}\xspace}

\newcommand{\quadEnc}{\ensuremath{\mathcal{E}^2}\xspace}
\newcommand{\quadEncN}{\ensuremath{\mathcal{E}^2}}
\newcommand{\cubicEnc}{\ensuremath{\mathcal{E}^3}\xspace}

\newcommand{\writeTO}{\ensuremath{\mathbf{WW}}\xspace}
\newcommand{\readWriteTO}{\ensuremath{\mathbf{RW}}\xspace}
\newcommand{\readFromTO}{\ensuremath{\mathbf{RF_{TO}}}\xspace}
\newcommand{\ppo}{\ensuremath{\mathbf{PPO}}\xspace}

\newcommand{\valueConstraintsCubi}{\ensuremath{\mathbf{RF}^3}\xspace}
\newcommand{\valueConstraintsCubii}{\ensuremath{\mathbf{FR}}\xspace}

\newcommand{\valueConstraintsQuadi}{\ensuremath{\mathbf{RF}^2}\xspace}
\newcommand{\valueConstraintsQuadii}{\ensuremath{\mathbf{SUP}}\xspace}

\newcommand{\memorylocation}{\langle\textit{ADDRESS}\rangle}

\newcommand{\Th}{\ensuremath{\mathcal{T}}\xspace}
\newcommand{\DPLLT}{DPLL(\Th)\xspace}

\newcommand{\atoms}{\ensuremath{\mathcal{A}}\xspace}
\newcommand{\TLiterals}{\ensuremath{\mathcal{L}_\atoms}\xspace}

\newcommand{\prop}[1]{\ensuremath{{#1}^{\mathbb{B}}}\xspace}
\newcommand{\TLemma}{\ensuremath{L}\xspace}
\newcommand{\TConf}{\ensuremath{\lnot\TLemma}\xspace}
\newcommand{\PLemma}[1]{\ensuremath{\prop{\TLemma}_{#1}}\xspace}
\newcommand{\PConf}[1]{\ensuremath{\lnot\PLemma{#1}}\xspace}
\newcommand{\Res}{\ensuremath{\textsc{Res}}\xspace}
\newcommand{\res}{\otimes}
\newcommand{\ThLearn}{\ensuremath{\Th\textsc{-learn}}\xspace}
\newcommand{\ThModels}{\models_{\Th}}

\newcommand{\Model}{\ensuremath{\mu}\xspace}
\newcommand{\VarAssign}{\ensuremath{\nu}\xspace}
\newcommand{\PreAbs}{\ensuremath{H}\xspace}
\newcommand{\assign}{\ensuremath{M}\xspace}
\newcommand{\Crits}{\ensuremath{Q}\xspace}

\newcommand{\bigand}{\bigwedge}
\newcommand{\bigor}{\bigvee}
\newcommand{\LInput}{\ensuremath{\alpha}}
\newcommand{\Last}{\ensuremath{\beta}}
\newcommand{\Proof}{\ensuremath{\Pi}\xspace}

\newcommand{\Map}{\ensuremath{\mathbf{F}}\xspace}
\newcommand{\MaxMap}{\ensuremath{\mathbf{F^*}}\xspace}

\newcommand{\PVars}{\ensuremath{\mathcal{X}}\xspace}
\newcommand{\emptyclause}{\Box}
\newcommand{\Lit}{\ell}

\newcommand{\prule}[3]{\infer[\ #1]{#3}{#2}}

\newcommand{\axiswidth}{17em}

\title{A Concurrency Problem with \\ Exponential DPLL(T) Proofs \\
  \small{Extended Version}}
\titlerunning{A Concurrency Problem with Exponential DPLL(T) Proofs}

\author{
  Liana Hadarean\inst{1}
  \and
  Alex Horn\inst{1}
  \and
  Tim King\inst{2}
}
\authorrunning{Hadarean, Horn, and King}

\institute{
  University of Oxford \\
  \email{liana.hadarean@cs.ox.ac.uk, alex.horn@cs.ox.ac.uk}
\and
  Verimag \\
  \email{tim.king@imag.fr}
}

\begin{document}

\maketitle

% Tim: I have found it useful to keep the abstract in a form that it can be copy pasted into text boxes in online forms. So try to not have line breaks or any special math symbols.
\begin{abstract}
Many satisfiability modulo theories solvers implement a variant of the \DPLLT framework which separates theory-specific reasoning from reasoning on the propositional abstraction of the formula. Such solvers conclude that a formula is unsatisfiable once they have learned enough theory conflicts to derive a propositional contradiction. However some problems, such as the diamonds problem, require learning exponentially many conflicts. We give a general criterion for establishing lower bounds on the number of theory conflicts in any \DPLLT proof for a given problem. We apply our criterion to two different state-of-the-art symbolic partial-order encodings of a simple, yet representative concurrency problem. Even though one of the encodings is asymptotically smaller than the other, we establish the same exponential lower bound proof complexity for both. Our experiments confirm this theoretical lower bound across multiple solvers and theory combinations.
\end{abstract}

\section{Introduction}
\label{section:intro}

Many high-level verification tools rely on satisfiability modulo theories (\SMT) solvers to discharge verification conditions in a variety of first-order logic theory theories. 
State-of-the-art \SMT solvers decide such problems by implementing variations on the \DPLLT framework. 
The \DPLLT framework integrates a theory-specific solver with efficient search over the propositional abstraction of the formula.
For this, \DPLLT uses a propositional (\SAT) solver that searches for a satisfying assignment to the propositional abstraction of the formula.
When such an assignment is found, a theory solver checks that this propositional assignment is theory consistent.
If it is not, a theory conflict (or $\Th$-conflict) clause is added, summarizing the inconsistency and preventing the \SAT solver from exploring this part of the search space again.
The process continues until either a theory consistent satisfying assignment is found, or a contradiction can be derived purely on the propositional level using the learned theory conflicts.
% If the input formula is unsatisifiable, the solver can derive a contradiction purely on the propositional level by taking into account 
While usually efficient in practice, there are well-known problems, such as the ``diamonds problem''~\cite{SSB2002}, on which the \DPLLT framework cannot derive a contradiction using a polynomial number of theory conflicts.
This issue has resurfaced in recent work on worst-case execution time~\cite{HAMM14}.
This limitation stems from the fixed alphabet of the \DPLLT theory conflicts.
Despite work on addressing this inherent inefficiency, it is still an open problem~\cite{BDdM08,TR2012}. 

In this paper, we prove a general theorem for establishing lower bounds on the number of $\Th$-conflicts in the \DPLLT calculus~\cite{NOT2006} required to prove that a given formula is unsatisfiable.
The theorem relies on the notion of \emph{non-interfering critical assignments}: propositionally satisfying assignments that contain disjoint $\Th$-conflicts. 
To the best of our knowledge, this is the first attempt at establishing a general framework for establishing lower bounds for \DPLLT proofs.
%
%The theorem provides a general rule for establishing lower bounds for \DPLLT proofs.
%To the best of our knowledge, this is the first proof of a well known ``folk-theorem'' in the SMT community.

We apply this theorem to study the \DPLLT proof complexity of proving a safety property of a simple, yet challenging concurrency problem.
The problem appears in the software verification competition (SV-COMP) and is of broad historical interest~\cite{OG1976,FKP2013}.
% LSH: we say this in section:problem-challenge
%This problem goes back to 1976 when it was used to illustrate the need for auxiliary variables in compositional proof rules~\cite{OG1976}, and most recently it has resurfaced as a challenge problem for automated verification tools~\cite{FKP2013}.
We focus on encodings recently implemented in a bounded model checker~\cite{AKT2013} because they have been successfully used to find concurrency-related bugs in software such as the Apache HTTP server, PostgreSQL and the Linux kernel~\cite{AKT2013}. Informally, these encodings symbolically model a certain partial-ordering between memory accesses, similar to the \emph{happens-before relations} in distributed systems~\cite{L1978}.

\vspace{-1em}
\paragraph{Contributions.} The main contributions of this paper are as follows: (1)~we give a new result for establishing lower bounds on the size of \DPLLT proofs of unsatisfiability; (2)~we propose a new problem challenge for the SMT community, whose solution is directly relevant to finding concurrency-related bugs in software; (3)~we establish a factorial lower bound on the size of \DPLLT proofs of unsatisfiability for this challenge problem; finally, (4)~we experimentally confirm the hardness of this problem.

\vspace{-1em}
\paragraph{Organization.} We prove the lower bound theorem in \autoref{section:noninterfering}. We introduce the problem challenge and explain how to generate two equisatisfiable partial-order encodings in \autoref{section:problem-challenge}. Given these encodings, we formalize the \DPLLT proof size complexity of the challenge problem (\autoref{section:lowerbounds}) and experimentally confirm its complexity (\autoref{section:experiments}).
We conclude with a discussion of related work and future research directions in \autoref{section:concl}.

%%  LocalWords:  satisfiability propositionally unsatisfiable

\section{Non-interfering Critical Assignments}
\label{section:noninterfering}

In this section, we give a general theorem for establishing lower bounds on the number of $\Th$-conflicts in all proofs that a formula $\phi$ is unsatisfiable in the \DPLLT calculus~\cite{NOT2006}.
The theorem is based on the notion of sets of \emph{non-interfering critical assignments} for $\phi$.

We assume readers are familiar with standard notions from SMT such as
$\Th$-conflicts, $\Th$-validity, $\Th$-lemmas, \DPLLT, etc.
In \DPLLT, a proof of unsatisfiability for a $\Th$-formula consists of a combination of learning $\Th$-valid lemmas and performing resolution steps on the propositional abstraction, until the empty clause is derived.
As in~\cite{NOT2006}, we restrict the proofs to work over the fixed alphabet $\atoms$ of $\Th$-atoms
in the input formula and that all $\Th$-lemmas are clauses.
We use a simplified view of the \DPLLT calculus~\cite{NOT2006} that only uses two rules: (i) propositional resolution (\Res) and (ii) learning $\Th$-valid clauses over the literals of $\atoms$ (\ThLearn).
We ignore $\Th$-propagation and splitting-on-demand~\cite{BNOT06}.

\vspace{-1em}
\paragraph{Notation.} 
We fix a set of propositional variables $\PVars$ and use $\Lit$ to denote literals over this set. 
A clause $C$ is a set of literals interpreted as their disjunction.
The empty clause $\emptyclause$ denotes false.
A \emph{partial assignment} $\assign$ is a set of literals that does not contain both a variable and its negation.
Partial assignments are interpreted as a conjunction
$\bigand_{\Lit \in \assign} \Lit$ and are always propositionally consistent.
An \emph{assignment} $\assign$ is a partial assignment s.t. for all $v \in \PVars$ either
$v \in \assign$ or $\lnot v \in \assign$.
The negation of a clause is a set of literals $\lnot C = \{\lnot \Lit \alt \Lit \in C\}$ and is interpreted as a conjunction.

The propositional abstraction function $\prop{\_}$ is an injective map from $\atoms$ into $\PVars$.
The $\Th$-literals, written $\TLiterals$, are the set of literals over $\atoms$.
We lift $\prop{\_}$ to work over $\Th$-literals and sets of $\Th$-literals.
We denote by $\TLemma$ a $\Th$-valid clause over $\TLiterals$, $\ThModels \bigor_{t \in L} t$,
and $\TConf$ will denote a $\Th$-conflict.
A $\Th$-conflict is a set of $\Th$-literals whose conjunction is $\Th$-unsatisfiable, $\TConf \ThModels \emptyclause$.
A \emph{minimal} $\Th$-conflict has the additional property that every strict subset is $\Th$-satisfiable.

\vspace{-1em}
\paragraph{Proofs.}
We assume the input $\Th$-formula $\phi$ has already been converted to CNF and is represented as a finite set of clauses
$C_1, \ldots C_\LInput$ over the variables in $\PVars$, the set of $\Th$-atoms $\atoms$, and the boolean abstraction function $\prop{\_} : \atoms \to \PVars$.
A Fixed-Alphabet-\DPLLT proof has the form:
\[
C_1, \ldots, C_\LInput, \ldots,  C_k, \ldots, C_{\Last} = \emptyclause
\]
where each $C_k$ for $\LInput < k \leq \Last$ is derived from a previous clause using either the resolution rule (\Res) or theory learning (\ThLearn).
% \footnote{
%   We are following the notation in~\cite{AKS10}.
% }
Let $C_i \res_{\Lit} C_j$ denote propositional resolution on $\Lit$. 
\[
\begin{array}{ll}
  \prule{\ThLearn}{ C_1, \ldots, C_k & \TLemma \subseteq \TLiterals & \ThModels \bigor_{t \in \TLemma } t }{C_1, \ldots, C_k, \prop{\TLemma} }  &
  \prule{\Res}{C_1, \ldots, C_k & 1 \le i < j \le k & \Lit \in C_i & \lnot \Lit \in C_j}{ C_1, \ldots, C_k, C_i \res_{\Lit} C_j} 
\end{array}
\]
The rule \ThLearn adds a new clause $\prop{\TLemma}$ that corresponds to the propositional
abstraction of a $\Th$-valid clause.
Clauses derived by \ThLearn are called $\Th$-lemmas.
\ThLearn is more general than Lazy Theory Learning~\cite{NOT2006},
which requires the literals to be in the partial assignment.

\vspace{-1em}
\paragraph{Critical Assignments.}
Given a \Th-formula $\phi$, an assignment $\assign$ is \emph{critical} if it satisfies the initial propositional abstraction of $\phi$ (i.e., $\assign \models \bigand_{i=1}^{\LInput} C_i$) and there is exactly one minimal $\Th$-conflict $\lnot \TLemma$ such that $\lnot \prop{\TLemma} \subseteq \assign$.
We denote by $\Crits$ a set of critical assignments for $\phi$, all of which can be enumerated as $\assign_1, \ldots, \assign_{|\Crits|}$ and where $\lnot\TLemma_i$ denotes the minimal $\Th$-conflict for $\assign_i$.
We say that $\Crits$ is \emph{non-interfering} whenever, for all $\assign_i \neq \assign_j$ in $\Crits$, $\PConf{i}$ is not a subset of $\assign_j$. In other words, no two assignments in $\Crits$ contain the same \Th-conflict. 

\begin{lemma}
  \label{lem:critical}
  Let $\assign$ be a critical assignment for $\phi$
  with the minimal $\Th$-conflict $\TConf$,
  and $\Proof$ be a Fixed-Alphabet-\DPLLT proof that $\phi$ is unsatisfiable.
  There is a \ThLearn application $C_k \in \Proof$ such that
  $\PConf{} \subseteq \lnot C_k \subseteq \assign$.
\end{lemma}
\begin{proof}
  The assignment $\assign$ does not satisfy the last clause $C_\Last = \emptyclause$ in $\Proof$.
  Therefore, there is some first clause $C_k$ that $\assign$ does not satisfy in $\Proof$.
  The clause $C_k$ cannot be an input clause as $\assign \models C_i$ for $1 \leq i \leq \LInput$.
  Additionally,
  $C_k$ cannot be the result of \Res: since $C_k$ is the first unsatisfied clause, 
  all $\assign \models C_i$ for $i < k$, and resolving  $C_i$ and $C_{i'}$
  for $i  \neq i' < k$ results in a clause satisfied by $\assign$.
  Thus $C_k$ must be the result of a \ThLearn application and $\assign \not\models C_k$.
%
%  We next show that any such $C_k$ must contain the conflict in the critical assignment $\assign$ not mapped by \MaxMap: $\PLemma{} \subseteq C_k$.
  Since $\assign$ is an assignment which does not satisfy $C_k$, $\assign$ must contain the negation of all literals in $C_k$. Equivalently, $\lnot C_k \subseteq \assign$.
  Let $T$ be the \Th-lemma corresponding to $C_k$: $C_k = \prop{T}$.
  As $\PConf{}$ is the unique minimal subset of $\assign$ that maps to a
  minimal theory conflict, $\TLemma \subseteq T$. Therefore, $\PConf{} \subseteq \lnot C_k \subseteq \assign$.
\end{proof}

Intuitively Lemma~\ref{lem:critical} states that, for each critical assignment $\assign$, the proof of unsatisfiability must contain a clause, derived by \ThLearn, which rules out \assign as a model of $\phi$ in the theory \Th.

\begin{theorem}
  \label{lem:noninterference}
  Let $\phi$ be an unsatisfiable $\Th$-formula, and
  let $\Crits$ be a non-interfering set of critical assignments for $\phi$.
  Then all Fixed-Alphabet-$\DPLLT$ proofs that $\phi$ is unsatisfiable
  contain at least $|\Crits|$ applications of \ThLearn.
\end{theorem}
\begin{proof}
  Let $\Proof$ be any Fixed-Alphabet-$\DPLLT$ proof.
  We will show that there exists a surjective partial map from \Th-lemmas in \Proof onto critical assignments in $\Crits$ that contain the same \Th-inconsistency. 
  We examine the set of partial maps \Map over $(\LInput,\Last]$ indices such that
  $\Map(k) = j$ only if $\PLemma{j} \subseteq C_k$ and $C_k$ is a \ThLearn application.
  Let the partial function \MaxMap be a partial function that maps
  onto the maximal number of distinct $\assign \in \Crits$ among all such maps \Map.
  If \MaxMap maps onto all elements in $\Crits$, there are at least $|\Crits|$
  applications \ThLearn in $\Proof$.
  If $|\Crits| = 0$, the property trivially holds on $\Proof$.

  For the remainder of this proof, assume that $|\Crits| \geq 1$.
  Suppose for contradiction that \MaxMap is not surjective.
  We can then select some critical assignment $\assign_j$ such that for all $k \in (\LInput,\Last]$
  either $k$ is not in the domain of $\MaxMap$ or $\MaxMap(k) \neq j$.

  % We first show that if this is the case, there must be some $C_k$ in $\Proof$ added by the
  % \ThLearn rule such that $\assign_j \not\models C_k$.
  % The assignment $\assign_j$ does not satisfy the last clause $C_\Last = \emptyclause$ in $\Proof$.
  % There thus is some first clause $C_k$ that $\assign_j$ does not satisfy in $\Proof$.
  % The clause $C_k$ cannot be an input clause as $\assign_j \models C_i$ for $1 \leq i \leq \LInput$ ($\assign_j$ is critical).
  % Additionally,
  % $C_k$ cannot be the result of \Res: since $C_k$ is the first unsatisfied clause, 
  % all $\assign_j \models C_i$ for $i < k$, and resolving  $C_i$ and $C_{i'}$
  % for $i,i' < k$ results in a clause satisfied by $\assign_j$.
  % Thus $C_k$ must be the result of a \ThLearn application and $\assign_j \not\models C_k$.

  By Lemma~\ref{lem:critical},
  there exists a \ThLearn application $C_k \in \Proof$ such that
  $\PConf{j} \subseteq \lnot C_k \subseteq \assign_j$.
  As $\PLemma{j} \subseteq C_k$, we know that it is possible for
  \MaxMap to map $C_k$ to some $\assign_m \in \Crits$.
  As \MaxMap is maximal and there is no conflict mapped to $\assign_j$,
  $\MaxMap(k) = m$ for some $m \neq j$.
  By the construction of \MaxMap, $\PLemma{m} \subseteq C_k$.
  Recall that $\lnot C_k \subseteq \assign_j$.
  Thus $\PConf{m} \subseteq \lnot C_k \subseteq \assign_j$.
  As $\assign_j$ contains both $\PConf{j}$ and $\PConf{m}$ for some distinct $\assign_m$ in $\Crits$,
  this contradicts the assumption that $\Crits$ is non-interfering.

  We can now conclude by contradiction that \MaxMap maps some
  clause that is the result of \ThLearn in $\Proof$ onto each
  $\assign \in \Crits$.
  Therefore $\Proof$ contains at least $|\Crits|$ applications of \ThLearn.
\end{proof}

There are many instances in the literature of \emph{diamond} benchmarks for which exponential
lower bounds on the number of $\Th$-conflicts have been given~\cite{SSB2002,BDdM08,MKS2009,AM2013,HAMM14}.
Theorem~\ref{lem:noninterference} can be seen as a generalization of the lower bound arguments for the diamond benchmarks.
The rest of this paper is devoted to a novel application of Theorem~\ref{lem:noninterference}.

%%  LocalWords:  disjunction unsatisfiable unsatisfiability

%%% Local Variables:
%%% mode: latex
%%% TeX-master: "paper"
%%% End:

\section{Challenge problem}
\label{section:problem-challenge}

%% \begin{SCfigure}[100][t]
%% \caption{A concurrent program $\texttt{T}_1 \parallel \ldots \parallel \texttt{T}_N$ where $0 < N$ is a fixed integer. We want to check the safety property $[x] \leq N$ if $[x]$ is initially $0$.
%% }
%% \label{fig:fkp2013}
%% \end{SCfigure}
In this section we present a challenge problem based on the \texttt{fpk2013} SV-COMP concurrency benchmark~\cite{FKP2013SVCOMP}.
This problem was first introduced in 1976 to illustrate the need for auxiliary variables in compositional proof rules for concurrent programs~\cite{OG1976}, and most recently it has resurfaced as a challenge problem for automated verification tools~\cite{FKP2013}.
Consider the following simple shared memory program with $N+1$ threads and a shared memory location $x$:
\vspace{-0.5em}
\begin{center}
\begin{tabular}{@{}l@{\hspace{2mm}} || l || c || @{\hspace{2mm}}l@{}}
%\begin{tabular}{@{}l@{\hspace{2mm}} || c || @{\hspace{2mm}}l@{}}
 \multicolumn{1}{c}{Thread $\texttt{T}_0$} & \multicolumn{1}{c}{Thread $\texttt{T}_1$} & \multicolumn{1}{c}{}      & \multicolumn{1}{c}{Thread $\texttt{T}_N$} \\
\midrule
 $\mathbf{local}\ v_0\ \texttt{:=}\ [x]$   & $\mathbf{local}\ v_1\ \texttt{:=}\ [x]$ & \multirow{2}{*}{$\ldots$} & $\mathbf{local}\ v_N\ \texttt{:=}\ [x]$   \\
 $\mathbf{assert}(v_0 \leq N)$                & $[x]\ \texttt{:=}\ v_1 + 1$             &                           & $[x]\ \texttt{:=}\ v_N + 1$
\end{tabular}
\end{center}
The memory at location $x$ is denoted by $[x]$. We assume that $[x]$ is initially $0$. Each thread $\texttt{T}_i$ reads the value at memory location $x$ into a CPU-local register $v_i$.
For $i \geq 1$, thread $\texttt{T}_i$ overwrites the memory at location $x$ with the new value $v_i + 1$.
For the rest of the paper, we denote the read of memory location $x$ in $\texttt{T}_0$ by $r_\mathit{assert}$.
The reads and writes on memory location $x$ in thread $\texttt{T}_i$ for $i \geq 1$ are denoted by $r_i$ and $w_i$, respectively.
%$r_k \deq \mathbf{local}\ v_k\ \texttt{:=}\ [x]$ and $w_k \deq [x]\ \texttt{:=}\ v_k + 1$ for $1 \leq k \leq 2$ denote the 
We follow the SV-COMP convention and assume sequential consistency~\cite{L1979}.
Therefore, if we just consider the concurrent program $\texttt{T}_1 \parallel \texttt{T}_2$,
we get the following six interleavings of shared memory accesses: (1)~$r_1; w_1; r_2; w_2$, (2)~$r_1; r_2; w_1; w_2$, (3)~$r_1; r_2; w_2; w_1$, (4)~$r_2; r_1; w_1; w_2$, (5)~$r_2; r_1; w_2; w_1$, (6)~$r_2; w_2; r_1; w_1$.
The different orders can result in different final values of $[x]$.
For example, $r_1; w_1; r_2; w_2$ results in the final value $2$ at memory location $x$, whereas $r_1; r_2; w_1; w_2$ results in the final value $[x] = 1$.

We want to check that the assertion $v_0 \leq N$ in thread $\texttt{T}_0$ cannot be violated.
%Equivalently, we want to check that $[x] \leq N$ is an invariant in the concurrent program $\texttt{T}_1 \parallel \ldots \parallel \texttt{T}_N$.
Intuitively, this assertion holds because each of the other $N$ threads increments $[x]$ at most once.
For a fixed $N$, we want to prove this automatically using bounded model checking.
While it is easy to automatically prove this property on each separate interleaving, the number of interleavings grows exponentially ($(2N + 1)! \div 2^N$). Next, %In~\autoref{section:partial-order-encoding},
we explain how to generate symbolic partial-order encodings that formalize all interleavings as a single quantifier-free SMT query.

%As artificial as this problem initially may appear to be, it has been of historical interest in concurrency verification, guiding research in the field.

%\section{Partial-order encodings}
%\label{section:partial-order-encoding}

\begin{figure}
\begin{align*}
%\bottom &\deq \bigand \set{ \clock{\bot} \clockpo \clock{e} \alt {e \in \events} } \text{ where } \forall e \in \events \colon \clock{\bot} \not= \clock{e} \\
\ppo &\deq \bigand \set{ (\guard{e} \land \guard{e'}) \Rightarrow (\clock{e} \clockpo \clock{e'}) \alt {e,e'\in \events \colon e\eventpo e'} } \\
\writeTO[x] &\deq \bigand \set{
  \parens{\clock{w} \clockpo \clock{w'} \lor \clock{w'} \clockpo \clock{w}} \land \writeSel{w} \neq \writeSel{w'} \alt {w, w' \in \writes_x \land w \neq w'} } \\
\readWriteTO[x] &\deq \bigand \left\{ \clock{w} \clockpo \clock{r} \lor \clock{r} \clockpo \clock{w} \alt {w \in \writes_x \land r \in \reads_x}\right\} \\
\readFromTO[x] &\deq \bigwedge \left\{\guard{r} \Rightarrow \bigvee \left\{ \writeSel{w} = \readSel{r} \alt {w \in \writes_x}\right\} \alt {r \in \reads_x}\right\}
\\[.75ex] \hline \\[-2.25ex]
\valueConstraintsCubi[x] &\deq \bigand \set{(\writeSel{w} = \readSel{r}) \Rightarrow \parens{ \guard{w} \land \writeval{w} = \readval{r} \land \clock{w} \clockpo \clock{r} } \alt {r\in\reads_x \land w\in\writes_x} } \\
\valueConstraintsCubii[x] &\deq \bigand \set{\parens{ \writeSel{w} = \readSel{r} \land \clock{w} \clockpo \clock{w'} \land \guard{w'} } \Rightarrow (\clock{r} \clockpo \clock{w'}) \alt {w,w'\in \writes_x \land r\in \reads_x} } \\
\cubicEnc &\deq \bigand \set{\readFromTO[x] \land \valueConstraintsCubi[x] \land \valueConstraintsCubii[x] \land \writeTO[x] \land \readWriteTO[x] \alt x \in \memorylocation } \land \ppo
\\[.75ex] \hline \\[-2.25ex]
\valueConstraintsQuadi[x] &\deq \bigwedge \left\{ (\writeSel{w} = \readSel{r}) \Rightarrow \left( \clock{w} = \supremum{r} \land \guard{w} \land \writeval{w} = \readval{r} \land \clock{w} \clockpo \clock{r} \right) \alt {r\in\reads_x \land w\in \writes_x}\right\} \\ 
\valueConstraintsQuadii[x] &\deq \bigwedge \left\{ \left(\clock{w} \clockpoeq \clock{r} \land \guard{w} \right) \Rightarrow (\clock{w} \clockpoeq \supremum{r}) \alt {r\in\reads_x \land w\in \writes_x}\right\} \\
\quadEnc &\deq \bigwedge \left\{ \readFromTO[x] \land  \valueConstraintsQuadi[x] \land \valueConstraintsQuadii[x] \land \writeTO[x] \land \readWriteTO[x] \alt x \in \memorylocation \right\} \land \ppo
\end{align*}
\vspace{-1.5em}
\caption{Given a shared memory program structure $P = \langle \events, \eventpo, \valf, \guardf \rangle$, \cubicEnc and \quadEnc encode $P$'s SC-relaxed consistency~\cite{HK2015} with a cubic and quadratic number of constraints, respectively.}
\label{fig:encoding}
\end{figure}

\vspace{-1em}
\paragraph{Partial-order encodings.}
We formalize two quantifier-free and equisatisfiable partial-order encodings of a concurrency semantics called SC-relaxed consistency~\cite{HK2015}: a cubic-sized encoding (\cubicEnc) and a quadratic-sized encoding (\quadEnc).
%Tim: Not sure this really helps the flow
%Both of these encodings will be used in~\autoref{section:experiments} to symbolically model the challenge problem from~\autoref{section:problem-challenge}.
The formula generated by each encoding is satisfiable if and only if the safety property in the shared memory program can be violated.

%This section is organized as follows: we first describe our assumptions and the shared memory program semantics. Next we describe the constraints common to the two encodings, and then focus on each individually.

% introduce SC-relaxed consistency and the assumptions we are making
%\paragraph{Assumptions.}
To get \cubicEnc and \quadEnc, we make four simplifying assumptions about the program \texttt{P} under scrutiny:
(i) \texttt{P}'s weak memory concurrency semantics equates to \emph{SC-relaxed consistency}~\cite{HK2015};
% LSH space saving
% which allows us to avoid focusing on the intricacies of some specific computer architecture (such as x86, ARM or POWER, e.g.~\cite{AMSS2012})
% or programming language.
%(such as C++11, e.g.~\cite{BOSSW2011}).
(ii) \texttt{P} is well-structured;
(iii) all loops in \texttt{P} have been unrolled so that the only remaining control-flow statements in \texttt{P} are if-then-else branches; finally,
(iv) every shared memory location accessed by \texttt{P} is known at compile-time. Avoiding these restrictions is beyond the scope of this paper that concerns itself with SMT solvers rather than program analysis techniques.

% shared memory program representation that we generate encodings on
The formulas generated by both encodings \cubicEnc and \quadEnc have three parts: (i)~\emph{clock constraints} that partially order memory accesses, similar to the happens-before relation in distributed systems~\cite{L1978}; (ii)~\emph{value constraints} that determine what values are read or written by the program if those clock constraints hold; and (iii)~\emph{selection constraints} that associate each read to a specific write event. Our symbolic partial-order encoding is therefore parameterized by three theories: \theoryClock for encoding the clock constraints, \theoryValue for encoding constraints on the symbolic program values, and \theorySelection for encoding selection constraints.
We assume that \theoryClock's signature includes strict and non-strict partial-order relations, denoted by \clockpo and \clockpoeq, respectively.
We also assume that \theoryValue's signature can encode a decidable fragment of common machine arithmetic such as bitvector or Presburger arithmetic.
\theorySelection is an uninterpreted theory.

\begin{definition}
  \label{def:shared-memory-program-structure}
  A \emph{shared memory program structure} is a tuple $P = \tuple{ \events, \eventpo, \valf, \guardf }$ where
  \events is a finite set of \emph{events},
  \eventpo is a partial order on \events,
  $\valf : \events \rightarrow \theoryValue$-terms and
  $\guardf : \events \rightarrow \theoryValue$-formulas.
  Let $\memorylocation$ be the set of memory locations.
  We assume that the set of events \events in $P$ can be partitioned into reads $\reads_x$ and writes $\writes_x$ on memory location $x \in \memorylocation$.
  Given an event $e$ in \events, let \clock{e} and $\readSel{e}$ be a \theoryClock-variable (\emph{clock variables}) and \theorySelection-variable (\emph{selection variables}), respectively.
For each read $r \in \reads$, let $\readval{r}$ be a unique \theoryValue-variable, called \emph{read variable}.
  The function $\valf$ maps a write event $w \in \writes$ to a $\theoryValue$-term $\writeval{w}$ built from read variables.
\end{definition}

The partial order \eventpo is the \emph{preserved program order} (PPO)~\cite{AMSS2012,AKT2013}.
The intuition behind PPO is that it determines which events cannot be reordered in any execution of the program. For sequentially consistent programs, the preserved program order corresponds to the order of instructions in each thread.
Note that $\tuple{\events,\eventpo}$ can be relaxed for weaker forms of consistency such as TSO, e.g.~\cite{AKT2013}. Intuitively, given an event $e$ in \events, $\guardf(e)$ denotes the necessary condition for $e$ to be enabled. The equality $\writeSel{w} = \readSel{r}$ in the theory \theorySelection means that a read event $r$ is `selected' so that its input value is equal to the output of a write event $w$. That is to say, when $\writeSel{w} = \readSel{r}$ holds, the \theoryValue-variable $\readval{r}$ is equal to the term $\writeval{w}$.

\begin{example}
\label{example:encoding}
The program described in~\autoref{section:problem-challenge} for $N=2$ corresponds to the following:
\begin{itemize}
\item $\events = \{ w_{init},r_1, w_1, r_2, w_2, r_{assert} \}$ is partitioned into $\reads_x = \{r_1, r_2, r_{assert}\}$ and $\writes_x = \{ w_{init}, w_1, w_2\}$ where $x \in \memorylocation$ is the concrete memory location accessed by threads $\texttt{T}_0$, $\texttt{T}_1$ and $\texttt{T}_2$.
\item According to PPO: $w_{init} \eventpo r_1 \eventpo w_1$,  $w_{init} \eventpo r_2 \eventpo w_2$, and $w_{init} \eventpo r_{assert}$. 
\item The $\valf$ function is defined as $\writeval{w_{init}} \deq 0$, $\writeval{w_1} \deq \readval{r_1} + 1$ and $\writeval{w_2} \deq \readval{r_2} + 1$.
  %Note that $\val{r_1}$, $\val{r_2}$ and $\val{r_{final}}$ are unconstrained because their values will be later determined by the constraints in our partial-order encodings.
\item Since the program has no if-then-else statements, $\guard{e} = \mathbf{true}$ for all events $e$ in $\events$.
\end{itemize}
\end{example}

Figure~\ref{fig:encoding} shows how to generate the cubic-size \cubicEnc and quadratic-size  \quadEnc partial-order encoding for a given shared memory program structure $P = \tuple{ \events, \eventpo, \valf, \guardf }$.
The first four formulas, $\ppo$, $\writeTO[x]$, $\readWriteTO[x]$, and $\readFromTO[x]$, are shared by \cubicEnc and \quadEnc.
The constraint \ppo encodes the preserved program order \eventpo.
%\bottom enforces a lower bound for all clocks.
The remaining constraints are with respect to some concrete memory location $x$. To model the information flow in the program, we encode a form of the \emph{read-from} relation~\cite{AMSS2012,AKT2013}. For a fixed memory location $x$ this relation defines a function from $\reads_x$ to $\writes_x$. We model this through the selection variables $\readSel{r}$ and \writeSel{w}, for each read $r \in \reads_x$ and write $w \in \writes_x$, together with the equality $\readSel{r} = \writeSel{w}$. The intuition is that the value of a write event $w \in \writes_x$ is observed by a read event $r \in \reads_x$ iff $\readSel{r} = \writeSel{w}$. The \readFromTO constraints ensures that at least one such equality holds for every read. \writeTO encodes that all writes on the same shared memory location are totally ordered in the happens-before relation and cannot have the same selection value, and \readWriteTO encodes that every read $r$ and write $w$ on the same shared memory location satisfy that $r$ happens-before $w$, or vice versa.
Note that if \clockpo is a total order, then \writeTO is equivalent to the clock and selection variables being distinct.
(In practice, the $\writeSel{w}$ variables are optimized out as distinct constants.)
The same is not true for \readWriteTO because two reads can have the same clock variables.

The main difference between \cubicEnc and \quadEnc is how they encode values being overwritten in memory. A read $r$ in $\reads_x$ can read from a write $w$ in $\writes_x$ if $w$ is the most recent write to $x$ that happens before $r$. In the case of \cubicEnc, this is encoded by \valueConstraintsCubii which corresponds to the `from-read' axiom~\cite{AMSS2012,AKT2013}, also known as the `conflict relation'~\cite{BDM2013}. This formula introduces a cubic number of constraints.
By contrast, \quadEnc encodes the \valueConstraintsQuadii constraint that requires only a quadratic number of constraints. For this, \valueConstraintsQuadii introduces a new variable \supremum{r} for every read $r$ in $\reads_x$ to encode the least upper bound (supremum) of all writes in $\writes_x$ that happen-before $r$. Since the set $\set{\clock{w} \alt w \in \writes_x}$, for all memory locations $x$, is totally ordered with respect to \clockpo in \theoryClock by $\writeTO[x]$, \supremum{r} is the maximum of all writes in $\writes_x$ that happen-before $r$ in $\reads_x$ according to \clockpo.
It was previously shown in~\cite[Theorem~4]{HK2015} that for a given shared memory program structure $P$ the formulas $\cubicEnc$ and $\quadEnc$ are equisatisfiable.

\section{Lower Bounds for Quadratic and Cubic Encodings}
\label{section:lowerbounds}
We show that the challenge problem from section~\ref{section:problem-challenge}
requires \DPLLT to enumerate at least $N!$ theory conflicts before it finds a proof of unsatisfiability, for either of the \cubicEnc or \quadEnc encoding where $N$ is the number of threads.

We begin by constructing a formula that encodes the challenge program using the \cubicEnc encoding.
As \cubicEnc is not directly in CNF, we perform the following simplifications in order to apply Theorem~\ref{lem:noninterference}:
(i) all of the guards $\guard{e}$ are ignored because they always evaluate to $\True$, and
(ii) implications are distributed across conjunctions in
the $\valueConstraintsCubi[x]$ constraints
[$A \imp (B \land C)$ iff $\parens{A \imp B} \land \parens{A \imp C}$].
We also assume that \clockpo is a total order in \theoryClock, and that \theoryValue is either bit-vector, Presburger, or real arithmetic.
We denote by $\Th$ the standard combined theory $\theoryClock + \theoryValue + \theorySelection$.
%(As $\theoryClock$, $\theoryValue$, and $\theorySelection$ are signature disjoint, we do not worry about theory combination.)
Figure~\ref{fig:phicube} shows the resulting quantifier-free \Th-formula, denoted by $\phi^3$. Note that $\phi^3$ is in CNF if we interpret implications in the obvious way.
\begin{figure}
\begin{align*}
  \phi^3 \equiv &
  \underbrace{\clock{w_{init}} \clockpo \clock{r_{assert}}}_{ \ppo}
  \land
  \underbrace{\bigand_{i=1 \ldots N} \clock{w_{init}} \clockpo \clock{r_i} \clockpo \clock{w_i}}_{\ppo}
  \land  
  \underbrace{\bigand_{w,w' \in \writes, w \neq w'} \clock{w} \neq \clock{w'} \land \writeSel{w} \neq \writeSel{w'}}_{\writeTO[x]}
  \land
  %% \underbrace{\bigand_{w,w' \in \writes} }_{\writeTO[x]}
  %% \land
  \underbrace{\bigand_{w \in \writes, r\in\reads} \clock{w} \neq \clock{r}}_{\readWriteTO[x]}
  \land
  \\ &
  \underbrace{\bigand_{w \in \writes, r \in \reads} (\writeSel{w} = \readSel{r}) \Rightarrow \clock{w} \clockpo \clock{r}}_{\valueConstraintsCubi[x]}
  \land
  \underbrace{\bigand_{r \in \reads} (\writeSel{w_{init}} = \readSel{r}) \Rightarrow 0 = \readval{r}}_{\valueConstraintsCubi[x]}
  \land
  \underbrace{\bigand_{i=1 \ldots N, r \in \reads} (\writeSel{w_i} = \readSel{r}) \Rightarrow \readval{r_i} + 1 = \readval{r}}_{\valueConstraintsCubi[x]}
  \\ &
  \underbrace{\bigand_{w,w' \in \writes, r \in \reads}
    \parens{\writeSel{w} = \readSel{r} \land \clock{w} \clockpo \clock{w'}}
    \Rightarrow \clock{r} \clockpo \clock{w'}}_{\valueConstraintsCubii[x]}
  \land  
  \underbrace{\bigand_{r \in \reads} \parens{\bigor_{w \in \writes} \writeSel{w} = \readSel{r} } }_{\readFromTO[x]}
  \land
  \underbrace{\readval{r_{assert}} > N}_{\mathbf{assert}(v_0 \leq N)}
\end{align*}
\vspace{-1.5em}
\caption{The \cubicEnc encoding for the challenge problem (when $\clockpo$ is total).}
\label{fig:phicube}
\end{figure}
%In Figure~\ref{fig:phicube}, the source of each group of constraints appears in the under brace.
Note that in the $\valueConstraintsCubi[x]$ constraints, each $\writeval{w}$ term  has been replaced by either $0$ or $\readval{r_i} + 1$.

Let $S_N$ be the set of all permutations over $[1,N]$. Consider the following sequence of events that can be constructed from the permutation function $\pi$ in $S_N$:
\[
\sigma(\pi): w_{init}, r_{\pi(1)}, w_{\pi(1)}, r_{\pi(2)}, w_{\pi(2)}, \ldots, r_{\pi(N)}, w_{\pi(N)}, r_{assert}.
\]
%For example, $\sigma(\pi) = w_{init}, r_{1}, w_{1}, r_{2}, w_{2}, \ldots, r_N, w_N, r_{assert}$ for $\pi = \tuple{1, 2, \ldots, N}$ in $S_N$. 
The run of $\sigma(\pi)$ corresponds to satisfying the following clock and selection constraints:
\[
\clock{w_{init}}
\clockpo \clock{r_{\pi(1)}}
\clockpo \clock{w_{\pi(1)}}
\clockpo \cdots
\clockpo \clock{r_{assert}}, \quad
\writeSel{w_{init}} = \readSel{r_{\pi(1)}}, \quad
\bigand_{i=1 \ldots N-1} \writeSel{w_{\pi(i)}} = \readSel{r_{\pi(i+1)}},\:\text{and} \quad
\writeSel{w_{\pi(N)}} = \readSel{r_{assert}}
\]
with distinct values for all $\writeSel{w}$ variables.
A first-order variable assignment $\VarAssign_\pi$ can be constructed to satisfy the above constraints.
(An explicit construction of $\VarAssign_\pi$ and proofs for Lemma~\ref{thm:phi:crit} and Theorem~\ref{thm:square:lb} are given in Appendix~\ref{app:proof}.)
For each \theoryClock or \theorySelection literal $\Lit$, we include $\prop{\Lit}$ in an assignment $\assign_\pi$ if $\Lit$ holds under $\VarAssign_\pi$.
%To make $\assign_\pi$ complete, $\assign_\pi$ must contain the abstraction of the \theoryValue literals.
%Let the $\Th$-conflict $\TConf_\pi$ be the following set of \theoryValue literals:
Consider the following $\theoryValue$-conflict:
\begin{equation*}
  \TConf_\pi =
  \set{\readval{r_{\pi(1)}} = 0}
  \cup
  \set{\readval{r_{\pi(i)}} + 1 = \readval{r_{\pi(i+1)}} \alt i = 1 \ldots N-1}
  \cup
  \set{\readval{r_{\pi(N)}} + 1 = \readval{r_{assert}}}
  \cup
  \set{\readval{r_{assert}} > N}.
\end{equation*}
Note that each $\Lit \in \TConf_\pi$ is unit-propagated by the \theoryClock and \theorySelection literals already in $\assign_\pi$ on the propositional abstraction of $\phi^3$.
We add $\PConf{\pi}$ to $\assign_\pi$.
The remaining \theoryValue equality atoms in $\phi^3$ are added negatively.
Now $\assign_\pi$ satisfies the propositional abstraction of $\phi^3$.
\begin{lemma}
  \label{thm:phi:crit}
  The assignment $\assign_\pi$ is a critical assignment for $\phi^3$ with the theory conflict $\TConf_\pi$.
\end{lemma}
  
\begin{theorem}
  \label{thm:cube:lb}
  All Fixed-Alphabet-\DPLLT proofs for $\phi^3$
  contain at least $N!$ applications of $\ThLearn$.
\end{theorem}
\begin{proof}
  Let $\Crits = \set{\assign_\pi \alt \pi \in S_N}$.
  For each pair of distinct $\pi$ and $\pi'$  in $S_N$,
  there is some adjacent pair of events with a different order in $\sigma(\pi)$ and $\sigma(\pi')$.
  Select $k$ so that $\tuple{r_{\pi(k)}, r_{\pi(k+1)}} \neq \tuple{r_{\pi'(k)}, r_{\pi'(k+1)}}$.
  The literal $\prop{(\readval{r_{\pi(k)}} + 1 = \readval{r_{\pi(k+1)}})}$
  is in  $\TConf_\pi$ and is not in $\assign_{\pi'}$.
  Thus $\PConf{\pi}$ is not a subset of $\assign_{\pi'}$, and $\Crits$ is non-interfering.
  The lemma follows directly from Theorem~\ref{lem:noninterference}.
\end{proof}

\begin{theorem}
  \label{thm:square:lb}
  Let $\phi^2$ be the \quadEnc encoding of the challenge problem.
  All Fixed-Alphabet-\DPLLT proofs that $\phi^2$ are unsatisfiable 
  contain at least $N!$ application of $\ThLearn$.
\end{theorem}

An important difference between the diamond benchmarks and this problem is that for diamonds it is reasonable to describe all minimal \Th-conflicts as they each also correspond to critical models.
For the \texttt{fkp} problem, the encoding is more complex, and there are other classes of $\Th$-conflicts.
The set $\Crits$ identifies those $\Th$-lemmas that \emph{must} appear during solving.

%%% Local Variables:
%%% mode: latex
%%% TeX-master: "paper"
%%% End:

%%  LocalWords:  unsatisfiability iff Presburger CNF init fkp
%%  LocalWords:  unsatisfiable

\section{Experiments}
\label{section:experiments}

In this section, we give experimental results that confirm the lower bounds on the \DPLLT proofs for the two encodings of the problem challenge (\autoref{section:problem-challenge}).
Our experiments are carried out along three dimensions: we use four SMT solvers (Boolector v2.0.6~\cite{boolector}, CVC4 2015-03-14~\cite{cvc4}, Yices v2.3.0~\cite{yices2}, and Z3 2015-03-29~\cite{z3}), and we evaluate both the cubic-size and quadratic-size encoding (\cubicEnc~and~\quadEnc) with respect to four different SMT-LIB theory combinations.

We performed all experiments on a 64-bit machine running GNU/Linux 3.16 with 2 Intel Xeon 2.5\,GHz cores and 4\,GB of memory.
The timeout for each individual benchmark is 1 hour.
Recall that \cubicEnc~and~\quadEnc are parameterized by three theories, \theoryClock, \theorySelection and \theoryValue.
We experiment with the theory of reals $\theoryReals$,
the theory of integers $\theoryInts$, and
the theory of bit-vectors $\theoryBV$.
In our experiments, we instantiate $\tuple{\theoryClock,\theorySelection, \theoryValue}$ to four configurations such that \theoryClock $=$ \theorySelection:
\begin{center}
\begin{tabular}{ll}
  (1) ``real-clocks-int-val'': $\tuple{\theoryReals,\theoryReals, \theoryInts}$, &
  (3) ``bv-clocks-int-val'': $\tuple{\theoryBV,\theoryBV, \theoryInts}$, and \\
  (2) ``real-clocks-bv-val'': $\tuple{\theoryReals,\theoryReals,\theoryBV}$, &
  (4) ``bv-clocks-bv-val'': $\tuple{\theoryBV,\theoryBV, \theoryBV}$.
\end{tabular}
\end{center}
CVC4 and Z3 were run on all benchmarks.
Boolector is only used on the fourth configuration, i.e. purely $\theoryBV$ benchmarks.
Yices was run on the ``real-clocks-int-val'' and ``bv-clocks-bv-val'' configurations.
We further distinguish between the SMT-LIB benchmarks by labelling them with \cubicEnc~or~\quadEnc.
For example, `real-clocks-bv-val-\cubicEnc' identifies benchmarks generated with the cubic encoding in which \theoryClock, \theorySelection and \theoryValue are respectively instantiated as $\theoryReals$, $\theoryReals$, and $\theoryBV$.

For all the ``*-bv-val'' benchmarks (except CVC4 for ``real-clocks-bv-val''), the solvers are essentially encoding the problem in propositional logic and using a SAT solver.\footnote{
  CVC4 was run with the flag \texttt{--bitblast=eager} on ``bv-clocks-bv-val'' benchmarks~\cite{HBJBT2014}.
}
The process of encoding into propositional logic (\emph{bit-blasting}) enables the solver to learn clauses not necessarily expressible in the original alphabet of the input atoms.
We therefore call these solver and configuration pairs \emph{bit-blasted combinations}. All other solver and configuration pairs are called \emph{\DPLLT combinations}. The \DPLLT combinations are the ``*-int-val'' configurations, and the run of CVC4 on ``real-clocks-bv-val''.\footnote{In this configuration CVC4 does not eagerly reduce \theoryBV to SAT.} 
\DPLLT combinations use Fixed-Alphabet-\DPLLT proofs, whereas bit-blasted combinations generally do not.

Given an instantiation of $\tuple{\theoryClock,\theorySelection, \theoryValue}$, we separately encode the \texttt{fkp2013-unsat} concurrency benchmarks with \cubicEnc~and~\quadEnc for all $N \in [3,9]$.
There are a total of $56$ different unsatisfiable SMT-LIB benchmarks.
The size of each benchmark depends on $N$ and whether we used \cubicEnc~or~\quadEnc. For example, for $N = 9$, the total number of symbolic expressions in \cubicEnc is 4085, whereas \quadEnc yields only 1604 symbolic expressions.

\begin{figure}
\centering
\begin{tikzpicture}
\begin{axis}[grid=major,
	width=.63*\textwidth,
	ymode=log,
	xtick=data,
	font=\scriptsize,
	%title=fkp2013-unsat,
	xlabel=Number of threads ($N$),
	ylabel=Number of SAT conflicts,
	]
\addplot[cyan, very thick, solid] table [x=N, y=conflicts, col sep=comma] {data-cvc4-fkp2013-real-clocks-int-values-sup-off-unsat-conflicts.csv};
\label{plot:cvc4-real-clocks-int-values-sup=off}
\addplot[cyan, thin, solid] table [x=N, y=conflicts, col sep=comma] {data-cvc4-fkp2013-real-clocks-int-values-sup-on-unsat-conflicts.csv};
\label{plot:cvc4-real-clocks-int-values-sup=on}
\addplot[red, very thick, solid] table [x=N, y=conflicts, col sep=comma] {data-cvc4-fkp2013-real-clocks-bv-values-sup-off-unsat-conflicts.csv};
\label{plot:cvc4-real-clocks-bv-values-sup=off}
\addplot[red, thin, solid] table [x=N, y=conflicts, col sep=comma] {data-cvc4-fkp2013-real-clocks-bv-values-sup-on-unsat-conflicts.csv};
\label{plot:cvc4-real-clocks-bv-values-sup=on}
\addplot[green, very thick, solid] table [x=N, y=conflicts, col sep=comma] {data-cvc4-fkp2013-bv-clocks-int-values-sup-off-unsat-conflicts.csv};
\label{plot:cvc4-bv-clocks-int-values-sup=off}
\addplot[green, thin, solid] table [x=N, y=conflicts, col sep=comma] {data-cvc4-fkp2013-bv-clocks-int-values-sup-on-unsat-conflicts.csv};
\label{plot:cvc4-bv-clocks-int-values-sup=on}
\addplot[blue, very thick, dotted] table [x=N, y=conflicts, col sep=comma] {data-cvc4-fkp2013-bv-clocks-bv-values-sup-off-unsat-conflicts.csv};
\label{plot:cvc4-bv-clocks-bv-values-sup=off}
\addplot[blue, thin, dotted] table [x=N, y=conflicts, col sep=comma] {data-cvc4-fkp2013-bv-clocks-bv-values-sup-on-unsat-conflicts.csv};
\label{plot:cvc4-bv-clocks-bv-values-sup=on}
\addplot[pink, very thick, solid] table [x=N, y=conflicts, col sep=comma] {data-z3-fkp2013-real-clocks-int-values-sup-off-unsat-conflicts.csv};
\label{plot:z3-real-clocks-int-values-sup=off}
\addplot[pink, thin, solid] table [x=N, y=conflicts, col sep=comma] {data-z3-fkp2013-real-clocks-int-values-sup-on-unsat-conflicts.csv};
\label{plot:z3-real-clocks-int-values-sup=on}
\addplot[gray, very thick, dotted] table [x=N, y=conflicts, col sep=comma] {data-z3-fkp2013-real-clocks-bv-values-sup-off-unsat-conflicts.csv};
\label{plot:z3-real-clocks-bv-values-sup=off}
\addplot[gray, thin, dotted] table [x=N, y=conflicts, col sep=comma] {data-z3-fkp2013-real-clocks-bv-values-sup-on-unsat-conflicts.csv};
\label{plot:z3-real-clocks-bv-values-sup=on}
\addplot[cyan, very thick, solid] table [x=N, y=conflicts, col sep=comma] {data-z3-fkp2013-bv-clocks-int-values-sup-off-unsat-conflicts.csv};
\label{plot:z3-bv-clocks-int-values-sup=off}
\addplot[cyan, thin, solid] table [x=N, y=conflicts, col sep=comma] {data-z3-fkp2013-bv-clocks-int-values-sup-on-unsat-conflicts.csv};
\label{plot:z3-bv-clocks-int-values-sup=on}
\addplot[brown, very thick, dotted] table [x=N, y=conflicts, col sep=comma] {data-z3-fkp2013-bv-clocks-bv-values-sup-off-unsat-conflicts.csv};
\label{plot:z3-bv-clocks-bv-values-sup=off}
\addplot[brown, thin, dotted] table [x=N, y=conflicts, col sep=comma] {data-z3-fkp2013-bv-clocks-bv-values-sup-on-unsat-conflicts.csv};
\label{plot:z3-bv-clocks-bv-values-sup=on}
\addplot[purple, very thick, solid] table [x=N, y=conflicts, col sep=comma] {data-yices-fkp2013-real-clocks-int-values-sup-off-unsat-conflicts.csv};
\label{plot:yices-real-clocks-int-values-sup=off}
\addplot[purple, thin, solid] table [x=N, y=conflicts, col sep=comma] {data-yices-fkp2013-real-clocks-int-values-sup-on-unsat-conflicts.csv};
\label{plot:yices-real-clocks-int-values-sup=on}
\addplot[lime, very thick, dotted] table [x=N, y=conflicts, col sep=comma] {data-yices-fkp2013-bv-clocks-bv-values-sup-off-unsat-conflicts.csv};
\label{plot:yices-bv-clocks-bv-values-sup=off}
\addplot[lime, thin, dotted] table [x=N, y=conflicts, col sep=comma] {data-yices-fkp2013-bv-clocks-bv-values-sup-on-unsat-conflicts.csv};
\label{plot:yices-bv-clocks-bv-values-sup=on}
\addplot[teal, very thick, dotted] table [x=N, y=conflicts, col sep=comma] {data-boolector-fkp2013-bv-clocks-bv-values-sup-off-unsat-conflicts.csv};
\label{plot:boolector-bv-clocks-bv-values-sup=off}
\addplot[teal, thin, dotted] table [x=N, y=conflicts, col sep=comma] {data-boolector-fkp2013-bv-clocks-bv-values-sup-on-unsat-conflicts.csv};
\label{plot:boolector-bv-clocks-bv-values-sup=on}
  \addplot[black, thin, solid] table [x=N, y=factorial, col sep=comma] {factorial-of-n-up-to-9.csv} node[right,pos=1] {$N!$};
\end{axis}
% Legend
    \matrix[
        matrix of nodes,
        anchor=south west,
        nodes={font=\scriptsize},
      ] at([xshift=22em]current axis.south west)
      {
\refentry{cvc4-real-clocks-int-val} & |[text width=\axiswidth]|cvc4-real-clocks-int-val-\{\cubicEnc, \quadEncN\} \\
\refentry{cvc4-real-clocks-bv-val} & |[text width=\axiswidth]|cvc4-real-clocks-bv-val-\{\cubicEnc, \quadEncN\} \\
\refentry{cvc4-bv-clocks-int-val} & |[text width=\axiswidth]|cvc4-bv-clocks-int-val-\{\cubicEnc, \quadEncN\} \\
\refentry{cvc4-bv-clocks-bv-val} & |[text width=\axiswidth]|cvc4-bv-clocks-bv-val-\{\cubicEnc, \quadEncN\} \\
\refentry{z3-real-clocks-int-val} & |[text width=\axiswidth]|z3-real-clocks-int-val-\{\cubicEnc, \quadEncN\} \\
\refentry{z3-real-clocks-bv-val} & |[text width=\axiswidth]|z3-real-clocks-bv-val-\{\cubicEnc, \quadEncN\} \\
\refentry{z3-bv-clocks-int-val} & |[text width=\axiswidth]|z3-bv-clocks-int-val-\{\cubicEnc, \quadEncN\} \\
\refentry{z3-bv-clocks-bv-val} & |[text width=\axiswidth]|z3-bv-clocks-bv-val-\{\cubicEnc, \quadEncN\} \\
\refentry{yices-real-clocks-int-val} & |[text width=\axiswidth]|yices-real-clocks-int-val-\{\cubicEnc, \quadEncN\} \\
\refentry{yices-bv-clocks-bv-val} & |[text width=\axiswidth]|yices-bv-clocks-bv-val-\{\cubicEnc, \quadEncN\} \\
\refentry{boolector-bv-clocks-bv-val} & |[text width=\axiswidth]|boolector-bv-clocks-bv-val-\{\cubicEnc, \quadEncN\} \\};
\end{tikzpicture}
\vspace{-1em}
\caption{Experimental results for the \texttt{fkp2013-unsat} benchmark using four SMT solvers and four SMT-LIB theory combinations. The graph shows the factorial growth of the number of SAT conflicts in both the cubic-size and quadratic-size partial-order encoding as $N$ increases.}

\label{fig:experiments}
\end{figure}
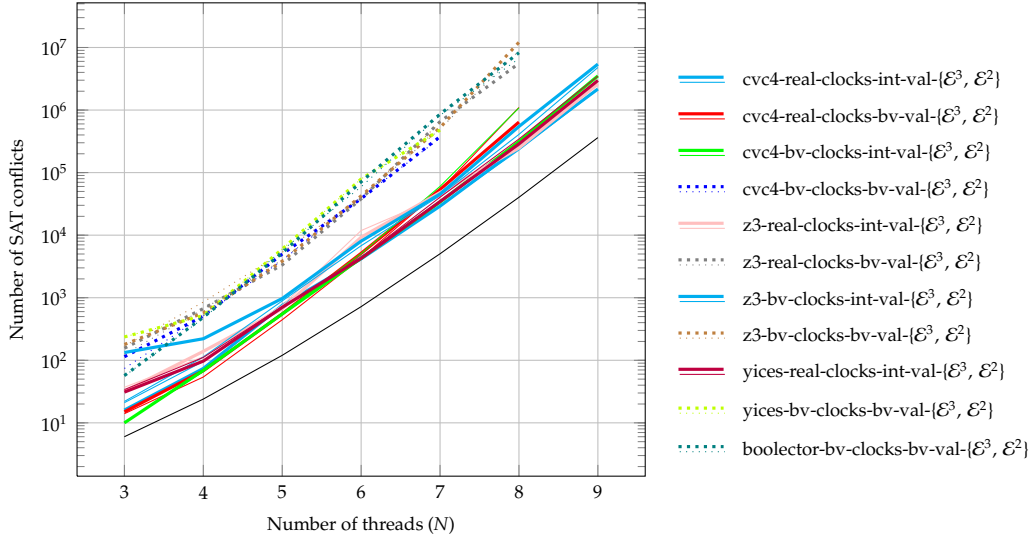

Figure~\ref{fig:experiments} charts the number of conflicts reported by each solver during execution.\footnote{
  Elapsed time and memory usage for the experiment is available in Appendix~\ref{sec:more:exp}.
}
Executions that exceeded the time limit of 1 hour are not included.
The $x$-axis corresponds to $N$.
The $y$-axis corresponds to the number of conflicts generated by the solver and has a logarithmic scale.
The legend for the chart groups together both the \cubicEnc (bold lines) and \quadEnc (thin lines) for a solver and theory specification.
These are further grouped into bit-blasted benchmarks (dotted lines) and \DPLLT (solid lines).
We also plot $N!$ as a black line.
The goal of the Figure~\ref{fig:experiments} is to convey the overall trends instead of compare individual data points.

% We next discuss our experimental results. The purpose of Figure~\ref{fig:experiments} is to summarize data trends (rather than individual data points). The $x$-axis and $y$-axis corresponds to the number of threads, $N$, and the number of SAT conflicts generated by the SMT solver, respectively. We group all plotted data points by \cubicEnc and \quadEnc where coloured bold (respectively, thin) lines, dots and dashes correspond to \quadEnc (respectively, \cubicEnc) benchmarks. Note that the $y$-axis is in logarithmic scale. This explains why the black curve for $N$ factorial is almost a straight line. Since all other data points lie above this black line, our experiments confirm that the number of SAT conflicts grows $\Omega(N!)$ for $3 \leq N \leq 9$.

We examine the number of SAT conflicts as it is a uniform and readily available statistic that is a lower bound on the number of proof steps taken by each solver.
Across all combinations, the number of conflicts observed is above the $N!$ line.
Thus  the $N!$ theory conflict lower bound proofs given in section~\ref{section:lowerbounds} holds for the \DPLLT combinations.
Our theoretical lower bounds do not extend to the bit-blasted combinations.
Nevertheless, our experiments show that the number of SAT conflicts are two orders of magnitude higher than $N!$ for bit-blasted combinations.
We therefore conjecture that a similar $N!$ lower bound exists for \Res proofs for the bit-blasted combinations.
We also examined CVC4's more detailed statistics on the \DPLLT combinations.
We confirmed that the number of \theoryValue-conflicts is always above $N!$ on the \DPLLT combinations.

\section{Conclusion}
\label{section:concl}

In this paper, we have demonstrated a theoretical factorial lower bound on the number of \ThLearn applications in all \DPLLT proofs for a challenge problem of historical interest using two state-of-the-art encodings.
Our encodings are most closely related to~\cite{AKT2013,HK2015}.
% where~\cite{AKT2013} shows that the cubic-size encoding can be used to find bugs in open-source software (including Apache HTTP server, PostgreSQL and the Linux kernel).
Experiments confirm the theoretical lower bound for \DPLLT proofs and show a strong relationship to the number of SAT conflicts in \Res-proofs for bitblasted bitvector encodings.
Both the theoretical relationships and the empirical relationships hold over a cubic \cubicEnc and a quadratic \quadEnc encoding.
Our experiments are therefore particularly significant for state-of-the-art tools such as CBMC (which implements a variant of \cubicEnc).
We believe that the kind of analysis we have undertaken throughout this paper provides an important diagnostic practice in the development of SMT encodings.
Future work will focus on handling the value constraints for partial-order encodings of weak memory concurrency and improving the performance of the SMT solvers on such benchmarks by moving outside of Fixed-Alphabet-\DPLLT proofs.

\paragraph{Acknowledgments}
Work partially supported by \href{http://erc.europa.eu/}{European Research Council} project~280053 (CPROVER)\href{http://www.cprover.org/}{\mbox{``CPROVER''}} and Grant Agreement nr.~306595 \href{http://stator.imag.fr/}{\mbox{``STATOR''}}.

\bibliographystyle{splncs03}
\bibliography{paper}

\appendix
\newpage

\section{Proofs for Lower Bounds}
\label{app:proof}
This section gives a more formal derivation for
$\VarAssign_\pi$ and $\assign_\pi$ the concepts discussed in
section~\ref{section:lowerbounds},
and proofs for Lemma~\ref{thm:phi:crit} and Theorem~\ref{thm:square:lb}.
We use $\Model, \VarAssign \models \phi$ to denote that a $\Sigma$-structure $\Model$ and a variable assignment over $\Model$ satisfies a $\Sigma$-formula $\phi$.

Let $\Model$ be any $(\theoryClock + \theorySelection + \theoryValue)$-structure
with the additional constraint that
\theoryClock, \theorySelection, and \theoryValue sorts are mapped to domains with cardinalities at least $|E|$, $N+1$, and $N+1$ respectively.
Such a structure exists unless \theoryValue is bit-vectors and the bit-width is insufficiently large.
We now construct a first-order variable assignment $\VarAssign_\pi$ over \theoryClock and \theorySelection variables to match $\sigma(\pi)$.
Let $x_1 \clockpo^\Model \cdots \clockpo^\Model x_{|E|}$ be any arbitrary chain in the \theoryClock domain of $\Model$, and let $\tuple{y_0, y_1, \ldots, y_N}$ be an arbitrary enumeration of $N+1$ distinct elements in the \theorySelection domain.
Both the $x_i$ chain and the $y_i$ sequence exist as the cardinalities are large enough.
We now assign the $\clock{e}$ and $\writeSel{e}$ variables.
\[
\VarAssign_\pi(\clock{e}) =
\begin{cases}
  x_1 & e = w_{init} \\
  x_{2i} & e = r_{\pi(i)} \\
  x_{2i+1} & e = w_{\pi(i)} \\
  x_{2N+2} & e = r_{assert}
\end{cases}
\qquad
\VarAssign_\pi(\writeSel{w})=
\begin{cases}
  y_0 & w = w_{init} \\
  y_{i} & w = w_{\pi(i)}
\end{cases}
\qquad
\VarAssign_\pi(\readSel{r})=
\begin{cases}
  y_{i} & r = r_{\pi(i+1)} \\
  y_{N} & r = r_{assert} \\
\end{cases}
\]

We construct a complete set of $\Th$-literals $\PreAbs_\pi$ (either $\Lit \in \PreAbs_\pi$ or $\lnot \Lit \in \PreAbs_\pi$ for all $\Lit \in \TLiterals$).
This will correspond to $\assign_\pi$ before abstraction.
%All of the $\Th$-atoms in $\phi^3$ are either \theoryClock, \theoryValue or \theorySelection. 
For any literal $\Lit$ over \theoryClock or \theorySelection atoms,
we evaluate $\Lit$ w.r.t. $\Model$ and $\VarAssign_\pi$ to assign it in $\PreAbs_\pi$,
i.e. $\Lit \in \PreAbs_\pi$ if $\Model,\VarAssign_\pi \models \Lit$.
For atoms over \theoryValue, we include the $\TConf_\pi$ literals in $\PreAbs_\pi$ (defined in section~\ref{section:lowerbounds}).
For all other \theoryValue equalities $\Lit$ in $\phi^3$, we include $\lnot \Lit \in \PreAbs_\pi$.
We now let $\assign_\pi = \prop{\PreAbs_\pi}$.

%% \begin{lemma*}[]
%%   Each $\assign_\pi$ is a critical assignment for $\phi^3$ with the conflict $\PConf{\pi}$.
%% \end{lemma*}
\begin{proof}[\textbf{Proof of Lemma~\ref{thm:phi:crit}}]
  Since $\TConf_\pi \subseteq \PreAbs_\pi$ and $\assign_\pi = \prop{\PreAbs_\pi}$, $\PConf{\pi} \subseteq \assign_\pi$.
  We now show that for each $\Lit \in \TConf_\pi$, we can extend $\VarAssign_\pi$
  to a new assignment $\VarAssign_\pi^\Lit$
  so that $\Model,\VarAssign_\pi^\Lit \models h$ for all $h \in \PreAbs_\pi \setminus \set{\Lit}$.
  For brevity, we denote by $\Lit_0 = \parens{\readval{r_{\pi(1)}} = 0}$,
  $\Lit_i = \parens{\readval{r_{\pi(i)}} + 1 = \readval{r_{\pi(i+1)}}}$ for $i \in 1 \ldots N-1$,
  $\Lit_{assert1} = \parens{\readval{r_{\pi(N)}} + 1 = \readval{r_{assert}}}$, and
  $\Lit_{assert2} = \parens{\readval{r_{assert}} > N}$.
  \begin{align*}
    \VarAssign_\pi^{\Lit_0}(\readval{r})
    &= \begin{cases}
      1 & r = r_{\pi(1)} \\
      j+1 & r = r_{\pi(j)} \\
      N+1 & r = r_{assert}
    \end{cases}
    &
    \VarAssign_\pi^{\Lit_{i}}(\readval{r})
    &= \begin{cases}
      0 & r = r_{\pi(1)} \\
      j & r = r_{\pi(j)}, j < i \\
      k+1 & r = r_{\pi(k)}, k \geq i \\
      N+1 & r = r_{assert}
    \end{cases}
    \\
    \VarAssign_\pi^{\Lit_{assert1}}(\readval{r})
    &= \begin{cases}
      0 & r = r_{\pi(1)} \\
      j & r = r_{\pi(j)} \\
      N+1 & r = r_{assert}
    \end{cases}
    &
    \VarAssign_\pi^{\Lit_{assert2}}(\readval{r})
    &= \begin{cases}
      0 & r = r_{\pi(1)} \\
      j & r = r_{\pi(j)} \\
      N & r = r_{assert}
    \end{cases}
  \end{align*}
  We omit ${\_}^\Model$ from the \theoryValue-constants $0,\ldots, N+1$ above.
  It is now that case that $\Model,\VarAssign_\pi^\Lit \models h$ for all $h \in \PreAbs_\pi \setminus \set{\Lit}$.
  Thus $\PreAbs_\pi \setminus \set{\Lit}$ is satisfiable modulo $\Th$.
  As every subset of $\PreAbs_\pi$ excluding exactly one literal in $\TConf_\pi$ is satisfiable modulo $\Th$, $\TConf_\pi$ is the unique minimal $\Th$-conflict in $\PreAbs_\pi$.
  Thus $\assign_\pi$ is a critical assignment.
\end{proof}

\begin{proof}[\textbf{Proof of Theorem~\ref{thm:square:lb}}]
  We extend $\VarAssign_\pi$ to assign $\supremum{r}$ to match $\sigma(\pi)$:
  $\VarAssign_\pi(\supremum{r_{\pi(1)}}) = \VarAssign_\pi(\clock{w_{init}})$,
  $\VarAssign_\pi(\supremum{r_{\pi(i)}}) = \VarAssign_\pi(\clock{w_{\pi(i-1)}})$, and
  $\VarAssign_\pi(\supremum{r_{assert}}) = \VarAssign_\pi(\clock{w_{\pi(N)}})$.
  We follow the same construction of $\PreAbs_\pi$, $\assign_\pi$, $\VarAssign_\pi^\Lit$, and $\Crits$ as before for $\phi^3$.
  $Q$ is a set of non-interfering critical assignments for $\phi^2$. 
\end{proof}

\newpage
\section{Time and Memory Usage}
\label{sec:more:exp}
Elapsed time and memory usage for \texttt{fkp2013-unsat} benchmark; TIMEOUT $=$ 1 hour.

\begin{center}
\scalebox{0.57}{%
\begin{tabular}{c||c||c}
\textbf{CVC4} & \textbf{Z3} & \textbf{Yices} and \textbf{Boolector} \\ \midrule
\begin{tabular}[t]{@{}r|r|r@{}}
$N$ & Time (s) & Memory (MB) \\ \midrule
&\multicolumn{2}{|c}{\textit{cvc4-real-clocks-int-val-\cubicEnc}} \\ \midrule
3 & 0.00 & 0.0 \\
4 & 0.00 & 0.0 \\
5 & 0.28 & 15.3 \\
6 & 1.74 & 18.0 \\
7 & 15.50 & 24.5 \\
8 & 200.28 & 81.0 \\
9 & 2718.45 & 557.3 \\ \midrule
&\multicolumn{2}{|c}{\textit{cvc4-real-clocks-int-val-\quadEnc}} \\ \midrule
3 & 0.00 & 0.0 \\
4 & 0.00 & 0.0 \\
5 & 0.30 & 14.4 \\
6 & 2.19 & 16.9 \\
7 & 18.79 & 22.4 \\
8 & 199.17 & 68.0 \\
9 & 2906.83 & 579.2 \\ \midrule
&\multicolumn{2}{|c}{\textit{cvc4-real-clocks-bv-val-\cubicEnc}} \\ \midrule
3 & 0.00 & 0.0 \\
4 & 0.00 & 0.0 \\
5 & 0.39 & 17.4 \\
6 & 3.39 & 21.3 \\
7 & 36.00 & 32.1 \\
8 & 512.64 & 147.1 \\
9 & TIMEOUT & 597.3 \\ \midrule
&\multicolumn{2}{|c}{\textit{cvc4-real-clocks-bv-val-\quadEnc}} \\ \midrule
3 & 0.00 & 0.0 \\
4 & 0.09 & 15.5 \\
5 & 0.59 & 16.6 \\
6 & 5.20 & 19.8 \\
7 & 56.09 & 32.2 \\
8 & 1277.46 & 274.5 \\
9 & TIMEOUT & 554.1 \\ \midrule
&\multicolumn{2}{|c}{\textit{cvc4-bv-clocks-int-val-\cubicEnc}} \\ \midrule
3 & 0.00 & 0.0 \\
4 & 0.00 & 0.0 \\
5 & 0.20 & 15.0 \\
6 & 1.40 & 17.8 \\
7 & 13.09 & 26.5 \\
8 & 141.88 & 80.7 \\
9 & 1811.85 & 642.8 \\ \midrule
&\multicolumn{2}{|c}{\textit{cvc4-bv-clocks-int-val-\quadEnc}} \\ \midrule
3 & 0.00 & 0.0 \\
4 & 0.00 & 0.0 \\
5 & 0.19 & 14.5 \\
6 & 1.99 & 17.6 \\
7 & 23.89 & 39.6 \\
8 & 600.84 & 359.4 \\
9 & TIMEOUT & 978.5 \\ \midrule
&\multicolumn{2}{|c}{\textit{cvc4-bv-clocks-bv-val-\cubicEnc}} \\ \midrule
3 & 0.00 & 0.0 \\
4 & 0.00 & 0.0 \\
5 & 0.20 & 22.4 \\
6 & 2.69 & 40.9 \\
7 & 116.99 & 280.2 \\
8 & TIMEOUT & 1911.2 \\
9 & TIMEOUT & 2022.0 \\ \midrule
&\multicolumn{2}{|c}{\textit{cvc4-bv-clocks-bv-val-\quadEnc}} \\ \midrule
3 & 0.03 & 15.2 \\
4 & 0.00 & 0.0 \\
5 & 0.19 & 20.1 \\
6 & 2.98 & 40.5 \\
7 & 188.26 & 307.8 \\
8 & TIMEOUT & 1801.0 \\
9 & TIMEOUT & 1785.3 \\
\end{tabular}
&
\begin{tabular}[t]{@{}r|r|r@{}}
$N$ & Time (s) & Memory (MB) \\ \midrule
&\multicolumn{2}{|c}{\textit{z3-real-clocks-int-val-\cubicEnc}} \\ \midrule
3 & 0.00 & 0.0 \\
4 & 0.00 & 0.0 \\
5 & 0.00 & 0.0 \\
6 & 0.99 & 17.1 \\
7 & 4.80 & 21.4 \\
8 & 37.89 & 28.7 \\
9 & 697.24 & 45.8 \\ \midrule
&\multicolumn{2}{|c}{\textit{z3-real-clocks-int-val-\quadEnc}} \\ \midrule
3 & 0.00 & 0.0 \\
4 & 0.00 & 0.0 \\
5 & 0.00 & 0.0 \\
6 & 1.39 & 18.3 \\
7 & 5.50 & 21.2 \\
8 & 50.99 & 28.5 \\
9 & 694.27 & 42.4 \\ \midrule
&\multicolumn{2}{|c}{\textit{z3-real-clocks-bv-val-\cubicEnc}} \\ \midrule
3 & 0.00 & 0.0 \\
4 & 0.09 & 15.7 \\
5 & 0.49 & 16.4 \\
6 & 7.19 & 19.6 \\
7 & 112.19 & 27.9 \\
8 & 1415.20 & 49.8 \\
9 & TIMEOUT & 68.5 \\ \midrule
&\multicolumn{2}{|c}{\textit{z3-real-clocks-bv-val-\quadEnc}} \\ \midrule
3 & 0.00 & 0.0 \\
4 & 0.10 & 15.7 \\
5 & 0.79 & 16.9 \\
6 & 8.29 & 20.3 \\
7 & 97.79 & 25.9 \\
8 & 2441.23 & 56.9 \\
9 & TIMEOUT & 62.6 \\ \midrule
&\multicolumn{2}{|c}{\textit{z3-bv-clocks-int-val-\cubicEnc}} \\ \midrule
3 & 0.00 & 0.0 \\
4 & 0.00 & 0.0 \\
5 & 0.18 & 16.9 \\
6 & 1.69 & 18.6 \\
7 & 13.98 & 22.9 \\
8 & 270.86 & 31.7 \\
9 & 1755.95 & 57.1 \\ \midrule
&\multicolumn{2}{|c}{\textit{z3-bv-clocks-int-val-\quadEnc}} \\ \midrule
3 & 0.00 & 0.0 \\
4 & 0.00 & 0.0 \\
5 & 0.29 & 16.5 \\
6 & 2.69 & 18.0 \\
7 & 26.67 & 21.2 \\
8 & 394.60 & 32.4 \\
9 & 2862.76 & 54.1 \\ \midrule
&\multicolumn{2}{|c}{\textit{z3-bv-clocks-bv-val-\cubicEnc}} \\ \midrule
3 & 0.00 & 0.0 \\
4 & 0.00 & 0.0 \\
5 & 0.10 & 14.8 \\
6 & 3.10 & 25.1 \\
7 & 70.59 & 40.8 \\
8 & 3521.01 & 145.8 \\
9 & TIMEOUT & 133.6 \\ \midrule
&\multicolumn{2}{|c}{\textit{z3-bv-clocks-bv-val-\quadEnc}} \\ \midrule
3 & 0.10 & 12.6 \\
4 & 0.10 & 13.1 \\
5 & 0.79 & 14.7 \\
6 & 10.59 & 26.1 \\
7 & 252.13 & 40.5 \\
8 & TIMEOUT & 101.8 \\
9 & TIMEOUT & 114.4 \\
\end{tabular}
&
\begin{tabular}[t]{@{}r|r|r@{}}
$N$ & Time (s) & Memory (MB) \\ \midrule
&\multicolumn{2}{|c}{\textit{yices-real-clocks-int-val-\cubicEnc}} \\ \midrule
3 & 0.00 & 0.0 \\
4 & 0.00 & 0.0 \\
5 & 0.00 & 0.0 \\
6 & 0.10 & 3.4 \\
7 & 1.30 & 3.6 \\
8 & 26.59 & 7.1 \\
9 & 1086.21 & 34.0 \\ \midrule
&\multicolumn{2}{|c}{\textit{yices-real-clocks-int-val-\quadEnc}} \\ \midrule
3 & 0.00 & 0.0 \\
4 & 0.00 & 0.0 \\
5 & 0.00 & 0.0 \\
6 & 0.10 & 3.0 \\
7 & 1.90 & 4.0 \\
8 & 38.49 & 8.3 \\
9 & 1416.26 & 40.9 \\ \midrule
&\multicolumn{2}{|c}{\textit{yices-bv-clocks-bv-val-\cubicEnc}} \\ \midrule
3 & 0.00 & 0.0 \\
4 & 0.00 & 0.0 \\
5 & 0.10 & 4.0 \\
6 & 3.89 & 5.5 \\
7 & 74.48 & 12.2 \\
8 & TIMEOUT & 70.9 \\
9 & TIMEOUT & 96.4 \\ \midrule
&\multicolumn{2}{|c}{\textit{yices-bv-clocks-bv-val-\quadEnc}} \\ \midrule
3 & 0.00 & 0.0 \\
4 & 0.00 & 0.0 \\
5 & 0.10 & 4.0 \\
6 & 3.90 & 5.6 \\
7 & 76.59 & 12.0 \\
8 & TIMEOUT & 74.3 \\
9 & TIMEOUT & 96.0 \\ \midrule
&\multicolumn{2}{|c}{\textit{boolector-bv-clocks-bv-val-\cubicEnc}} \\ \midrule
3 & 0.00 & 0.0 \\
4 & 0.10 & 4.3 \\
5 & 0.69 & 5.6 \\
6 & 5.39 & 10.0 \\
7 & 94.29 & 34.2 \\
8 & 1491.46 & 95.2 \\
9 & TIMEOUT & 140.7 \\ \midrule
&\multicolumn{2}{|c}{\textit{boolector-bv-clocks-bv-val-\quadEnc}} \\ \midrule
3 & 0.00 & 0.0 \\
4 & 0.09 & 4.6 \\
5 & 0.69 & 6.1 \\
6 & 4.30 & 9.6 \\
7 & 86.49 & 29.1 \\
8 & 1122.07 & 87.4 \\
9 & TIMEOUT & 133.8 \\
\end{tabular}
\end{tabular}\vspace*{0.2cm}
%\label{table:time-and-memory}
}
\end{center}

\end{document}